\documentclass[journal,lettersize]{IEEEtran}
\IEEEoverridecommandlockouts

\usepackage{multicol}
\usepackage{amssymb}
\usepackage{amsmath}
\usepackage{amsfonts}
\usepackage{amsthm}
\usepackage{algorithm,algpseudocode}
\usepackage{bm}
\usepackage{caption}
\usepackage[noadjust]{cite}
\usepackage{chemarrow}
\usepackage{graphicx}
\usepackage{textcomp}
\usepackage{xcolor}
\usepackage{siunitx}
\usepackage[nolist,nohyperlinks]{acronym}
\usepackage{hyperref}
\usepackage{tikz}
\usepackage{mathtools}
\usepackage{multirow}
\usepackage{colortbl}

\newtheorem{theorem}{Theorem}

\acrodef{BNC}{bacterial nanocellulose}
\acrodef{DDS}{drug delivery system}
\acrodef{MC}{molecular communication}
\acrodef{PBS}{particle-based simulation}
\acrodef{wlog}[w.l.o.g.]{without loss of generality}

\newcommand{\Dc}{D_\mathrm{c}}

\allowdisplaybreaks

\makeatletter
\long\def\@makecaption#1#2{\ifx\@captype\@IEEEtablestring%
    \footnotesize\begin{center}{\normalfont\footnotesize #1}\\
        {\normalfont\footnotesize\scshape #2}\end{center}%
    \@IEEEtablecaptionsepspace
    \else
    \@IEEEfigurecaptionsepspace
    \setbox\@tempboxa\hbox{\normalfont\footnotesize {#1.}~~ #2}%
    \ifdim \wd\@tempboxa >\hsize%
    \setbox\@tempboxa\hbox{\normalfont\footnotesize {#1.}~~ }%
    \parbox[t]{\hsize}{\normalfont\footnotesize \noindent\unhbox\@tempboxa#2}%
    \else
    \hbox to\hsize{\normalfont\footnotesize\hfil\box\@tempboxa\hfil}\fi\fi}
\makeatother

\makeatother

\begin{document}
\bstctlcite{IEEEexample:BSTcontrol}

\title{Molecular Communication for Gastroretentive Drug Delivery
}
\author{Sebastian~Lotter,
        Marco~Seiter,
        Maryam~Pirmoradi,
        Lukas~Brand,
        Dagmar~Fischer,\\
        Robert~Schober
\thanks{Sebastian Lotter, Maryam Pirmoradi, Lukas Brand, and Robert Schober are with the Institute for Digital Communications, and Marco Seiter and Dagmar Fischer are with the Division of Pharmaceutical Technology and Biopharmacy, both at Friedrich-Alexander-Universität Erlangen-Nürnberg, Erlangen, Germany.}
\thanks{This work was funded in part by the Deutsche Forschungsgemeinschaft (DFG, German Research Foundation) – GRK 2950 – Project-ID 509922606.}
}

\maketitle
\begin{abstract}
Recently, \ac{BNC}, a biological material produced by non-pathogenic bacteria that possesses excellent material properties for various medical applications, has received increased interest as a carrier system for drug delivery.
However, the vast majority of existing studies on drug release from \ac{BNC} are feasibility studies with modeling and design aspects remaining largely unexplored.
To narrow this research gap, this paper proposes a novel model for the drug release from \ac{BNC}.
Specifically, the drug delivery system considered in this paper consists of a \ac{BNC} fleece coated with a polymer.
The polymer coating is used as an additional diffusion barrier, enabling the controlled release of an active pharmaceutical ingredient.
The proposed physics-based model reflects the geometry of the \ac{BNC} and incorporates the impact of the polymer coating on the drug release.
Hence, it can be useful for designing \ac{BNC}-based drug delivery systems in the future.
The accuracy of the model is validated with experimental data obtained in wet lab experiments.
\end{abstract}
\acresetall
\section{Introduction}\label{sec:introduction}
When medical drugs are administered to patients, it is of utmost importance that these drugs act as specifically as possible on their respective targets while causing the least possible amount of side effects on the healthy parts of the body.
The efficacy of a medication hereby depends on the ratio of drugs that reach the specific target site and on how well the drug concentration at the target site matches the therapeutic requirements.
The research field of {\em controlled release} targets the latter aspect by investigating methods to control the drug release rate from a drug carrier system, such as a tablet, a nanocapsule, or a hydrogel, so that eventually the drug is supplied to the target site at the therapeutically optimal rate \cite{adepu_controlled_2021}.

Drug delivery and specifically the controlled release of drugs from different types of drug carriers has been a research focus of \ac{MC} for many years \cite{chude-okonkwo_molecular_2017}.
Aspects of drug delivery being researched in the context of \ac{MC} include mathematical modeling of the drug transport in the human vascular system \cite{chahibi_molecular_2013,chen_modeling_2017,chude-okonkwo_information-theoretic_2020}, theoretical models for optimizing the local release rate of drugs at a target site \cite{zhao_release_2021,femminella_molecular_2015,salehi_diffusion-based_2019,zhao_adaptive_2021,sun2025bio,wang2021optimal}, and the use of extracellular vesicles for drug delivery \cite{veletic_modeling_2020,damrath_optimization_2024,rudsari_targeted_2021}.
While early studies in \ac{MC} in general and in \ac{MC}-based drug delivery in particular were mostly theoretical \cite{chude-okonkwo_molecular_2017}, experimental confirmation of the proposed theoretical models and system designs has increasingly gained importance in recent years \cite{lotter_experimental_2023}.
However, most \ac{MC} concepts for controlled release proposed in the literature lack experimental validation.
In this paper, we undertake a joint theoretical--experimental effort to narrow this research gap and investigate drug release from a promising carrier system: {\em \ac{BNC}}.

\ac{BNC} is a highly pure, biocompatible, and mechanically robust biomaterial produced by certain bacteria \cite{potzinger_bacterial_2017,iguchi_bacterial_2000}.
Its nanofibrous network structure, high water-holding capacity, and ease of functionalization have led to successful applications in wound dressings, tissue engineering, and, more recently, drug delivery.
Building on this, we explore the concept of coating \ac{BNC} fleeces with thin polymer layers. The coating functions as an additional diffusion barrier, which is stable under acidic conditions and enables extended drug release in the gastrointestinal tract, especially in the stomach.

Polymer-coated \ac{BNC} fleeces thus represent a novel and highly promising drug carrier for the treatment of gastric diseases such as ulcers and tumors.
However, while some preliminary investigations on the drug release from \ac{BNC} exist \cite{lotter_experimental_2023}, the mechanisms and kinetics of drug release from the coated \ac{BNC}-based systems considered in this paper are not yet understood and, to the best of the authors' knowledge, no corresponding mathematical model has been proposed in the literature to date.
A deeper understanding gained from mathematical modeling of these processes is crucial for optimizing the geometry of the \ac{BNC} and the material properties of both the fleece and its coating for specific therapeutic needs.

The main contributions of this work are as follows:
\begin{itemize}
    \item We develop a physics-based analytical model for drug release from polymer-coated \ac{BNC} fleeces formulated as a boundary value problem.
    \item We derive an explicit analytical solution that enables efficient computation of the release profile.
    \item We validate the proposed model using experimental data from wet-lab drug release studies.
\end{itemize}

By linking measurable coating parameters to resulting release profiles through a validated, analytical, and computationally efficient model, this work provides a practical tool for guiding the design of \ac{BNC}-based drug delivery systems.
Such a model-driven approach enables rapid prediction of release behavior for different geometries, coating properties, and drug types, thereby reducing experimental workload and facilitating the targeted development of gastroretentive systems with desired therapeutic profiles.

The remainder of this letter is organized as follows.
Section~\ref{sec:system_model} introduces the system model and the boundary value problem corresponding to the drug release.
In Section~\ref{sec:analysis}, an explicit solution to the boundary value problem is presented.
In Section~\ref{sec:evaluation}, numerical results obtained with the proposed model are compared to empirical data from wet-lab experiments.
Finally, Section~\ref{sec:conclusion} concludes the paper by summarizing the main findings and providing a brief outlook.
\section{Experimental Setup and System Model}\label{sec:system_model}
\subsection{Experimental Setup}\label{sec:experimental_setup}
All experiments were conducted in the Fischer lab. 
Cylindrical \ac{BNC} fleeces were loaded with pramipexole dihydrochloride monohydrate and freeze-dried.
Pramipexole is a drug that modulates dopaminergic signaling, i.e., chemical signaling pathways based on the neurotransmitter dopamine, in the central nervous system and is used in the therapy of Parkinson's disease and the restless legs syndrome.

The freeze-dried fleeces were coated with a polymer dispersion.
Figure~\ref{fig:coated_BNC} shows the cross section of a \ac{BNC} fleece after coating, where \ac{BNC} core (white) and polymer coating (pink) are clearly distinguishable.
The duration of the coating process was varied in the experiments to produce coatings of different thicknesses. The successful coating with shells of different thicknesses was confirmed by measuring the increase in mass after the coating process.
Finally, the drug release was measured in a shaking incubator at $37^{\circ}$ Celsius and $50$ revolutions per minute (rpm) using $0.1$ $\mathrm{mol}/\mathrm{l}$ hydrochloric acid solution (pH $1.2$) as release medium.
As a baseline, the drug release from uncoated \ac{BNC} fleeces was also measured under the same conditions. A detailed assessment of the experimental setup is provided in Section~\ref{sec:evaluation}.

\begin{figure}
    \centering
    \includegraphics[width=0.79\linewidth]{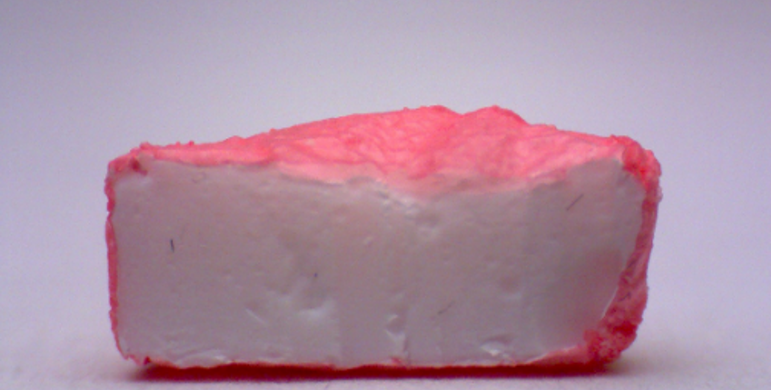}
    \caption{Drug-loaded BNC fleece (white) coated with a polymer (pink) that is resistant to degradation in the gastric environment. The image shows a cross-section of the coated fleece.}
    \label{fig:coated_BNC}
\end{figure}

\subsection{System Model}\label{sec:system_model_intro}

The geometry of the considered \ac{BNC} fleece is approximated as the cylindrical domain $\Omega = \lbrace (r,\phi,z) \,\vert\, 0 \leq r \leq R, 0 \leq \phi < 2\pi,  0 \leq z \leq Z \rbrace$, where $r$, $\phi$, and $z$ denote the radial, angular, and axial coordinate, respectively, and $R$ and $Z$, respectively, denote radius and height of the fleece.
The diffusive transport of drug molecules in $\Omega$ is governed by the following diffusion equation
\begin{align}
    \frac{\partial c(r,\phi,z,t)}{\partial t} = D \Delta c(r,\phi,z,t), \, (r,\phi,z) \in \Omega, \,t \geq 0, \label{eq:system_model:diff_eq}
\end{align}
where $c(r,\phi,z,t)$ denotes the concentration of drug molecules at time $t$ at $(r,\phi,z)$, $D$ the diffusion coefficient, and $\Delta$ the Laplacian in cylindrical coordinates, i.e.,
\begin{align}
\Delta c = \frac{1}{r}\frac{\partial}{\partial r}\left(r \frac{\partial c}{\partial r}\right) + \frac{1}{r^2}\frac{\partial^2 c}{\partial \phi^2} + \frac{\partial^2 c}{\partial z^2},
\end{align}
where here and in the following we suppress the dependency of $c(r,\phi,z,t)$ on space and time in the notation for the sake of readability wherever this does not lead to potential ambiguities.

In line with the experimental setup considered in this paper, we assume that the \ac{BNC} fleece is uniformly loaded with drug molecules, inducing the following initial condition
\begin{align}
    c(r,\phi,z,0) = c_0, \label{eq:system_model:initial_condition}
\end{align}
where $c_0$ denotes the initial drug concentration in $\Omega$ and \ac{wlog} we consider normalized drug concentration and release, i.e., $c_0 = 1/(Z R^2 \pi)$ in the following.
Furthermore, the impact of the polymer coating of the \ac{BNC} fleece on the drug release is modeled by the following boundary conditions
\begin{align}
    -D \left.\frac{\partial c}{\partial r}\,\right\vert_{r=R} = h \, c(R,\phi,z,t), \label{eq:system_model:radial_boundary}\\
    D \left.\frac{\partial c}{\partial z}\,\right\vert_{z=0} = h \, c(r,\phi,0,t), \label{eq:system_model:axial_boundary_down}\\
    -D \left.\frac{\partial c}{\partial z}\,\right\vert_{z=Z} = h \, c(r,\phi,Z,t), \label{eq:system_model:axial_boundary_up}
\end{align}
where $h = \Dc/l$ denotes the modified Sherwood number \cite{jain_theoretical_2022}, and $l$ and $\Dc$ denote the thickness of the polymer coating and the diffusion coefficient of the drug molecules in the coating, respectively.
Eqs.~\eqref{eq:system_model:radial_boundary}-\eqref{eq:system_model:axial_boundary_up} follow from the so-called {\em thin layer} approximation, which represents the impact of the relatively thin coating on the drug release as diffusive permeability $h$ \cite{jain_theoretical_2022}.
Finally, the following boundary condition completes the system model by ensuring that the drug concentration on the axis of $\Omega$ is finite
\begin{align}
    \lim_{r\to0} \,r\,c = 0.\label{eq:system_model:inner_bd_condition}
\end{align}

With \eqref{eq:system_model:initial_condition}, \eqref{eq:system_model:radial_boundary}, \eqref{eq:system_model:axial_boundary_down}, and \eqref{eq:system_model:axial_boundary_up}, angular symmetry of $c$ is implied, i.e., $\frac{\partial c}{\partial \phi} = 0$ everywhere in $\Omega$ and for all $t \geq 0$.
Hence, \eqref{eq:system_model:diff_eq} simplifies to
\begin{align}
    \frac{1}{D}\frac{\partial c}{\partial t} = \frac{1}{r}\frac{\partial}{\partial r}\left(r \frac{\partial c}{\partial r}\right) +  \frac{\partial^2 c}{\partial z^2}.\label{eq:system_model:diff_eq_simp}
\end{align}

\section{Analysis}\label{sec:analysis}
The solution to the boundary value problem \eqref{eq:system_model:diff_eq_simp},  \eqref{eq:system_model:initial_condition}-\eqref{eq:system_model:inner_bd_condition}, can be readily obtained by applying the method of separation of variables \cite{hahn_heat_2012}.
In particular, by decomposing \eqref{eq:system_model:diff_eq_simp} into eigenvalue problems in $r$, $z$, and $t$, respectively, and applying the Sturm-Liouville theory, the solution is obtained in the form of a series expansion in the respective orthogonal eigenfunctions in $r$ and $z$.
The following theorem summarizes the result.
\begin{theorem}\label{thm:concentration}
    The solution to \eqref{eq:system_model:diff_eq_simp},  \eqref{eq:system_model:initial_condition}-\eqref{eq:system_model:inner_bd_condition} is given as 
    \begin{align}
        &c(r,z,t) = \frac{1}{ZR^2 \pi}\sum_{n=1}^\infty \sum_{m=1}^\infty A_{nm} J_0(\alpha_n r) \exp(-D(\alpha_n+\beta_m)^2t)\nonumber\\ &\quad \left[\beta_m \cos(\beta_m z)+ (h/D) \sin(\beta_m z)\right],\label{eq:thm:conc}
    \end{align}
    where parameters $h$, $Z$, and $R$ have been defined in Section~\ref{sec:system_model_intro} and $\alpha_n = \gamma_n/R$ and $\gamma_n$ are the solutions of the following equation
    \begin{align}
        \gamma J_1(\gamma) = \frac{h R}{D} J_0(\gamma),\label{eq:thm:radial_eigv}
    \end{align}
    where $J_\nu(\cdot)$ denotes the Bessel function of the first kind of order $\nu$, $\beta_m$ are the positive solutions of the transcendental equation
    \begin{align}
        \tan(\beta Z) = \frac{2 h \beta}{D(\beta^2-h^2/D^2)},\label{eq:thm:axial_eigv}
    \end{align}
    and the coefficients $A_{nm}$ are obtained from \eqref{eq:system_model:initial_condition} as
    \begin{align}
        \!\!A_{nm} \!=\! \frac{\frac{R}{\alpha_n} \!J_1(\alpha_n R) \!\!\left[\frac{h}{D \beta_m}\! + \! \sin(\beta_m Z) \!- \!\frac{h}{D \beta_m}\!\cos(\beta_m Z)\right]\! }{N(\alpha_n)M(\beta_m)}\!,\!\!\label{eq:thm:coefficients}
    \end{align}
    where
    \begin{align}
        N(\alpha_n) &=\! \frac{J_0^2(\alpha_n R)}{2} \frac{R^2 \left[\left(\frac{h}{D}\right)^2+ \alpha_n^2\right]}{\alpha_n^2},\label{eq:thm:radial_norm}\\
        M(\beta_m)\! &=\! \!\frac{1}{2} \!\left[\!\left(\!\beta_m^2 \!+\! \left(\frac{h}{D}\right)^2\right)\!\!\left(\!Z \!+\! \frac{\frac{h}{D}}{\beta_m^2+\left(\frac{h}{D}\right)^2}\!\right)\!+\!\frac{h}{D}\right]\!.\!\label{eq:thm:axial_norm}
    \end{align}
\end{theorem}
\begin{proof}
The proof of the theorem is found in the appendix.
\end{proof}

From Theorem~\ref{thm:concentration}, the normalized cumulative drug release over the entire surface of the coated \ac{BNC} fleece, $F(t)$, is readily derived in terms of the corresponding fluxes over the boundaries $r=R$, $z=0$, and $z=Z$ as
\begin{align}
    F(t) &= h \left[\int\limits_{0}^{t} \int\limits_{0}^{R} 2 \pi r \left(c(r,Z,\tau)+c(r,0,\tau)\right) \,\mathrm{d}r\mathrm{d}\tau\right.\nonumber\\
    &\quad {}+\left.\int\limits_{0}^{t} \int\limits_{0}^{Z} 2 \pi R\,c(R,z,\tau) \,\mathrm{d}z\mathrm{d}\tau\right],\label{eq:analysis:total_flux_int}
\end{align}
where the definitions of the respective fluxes follow from \eqref{eq:system_model:radial_boundary}-\eqref{eq:system_model:axial_boundary_up}.
After computing the integrals on the right-hand side of \eqref{eq:analysis:total_flux_int} analytically, the simplified expression \eqref{eq:analysis:total_flux}, displayed on top of the next page, results.
\begin{figure*}
\centering
\begin{align}
    F(t) &= \frac{2h}{ZR^2} \sum_{n=1}^\infty \sum_{m=1}^\infty A_{nm} \left[\frac{R}{\alpha_n} J_1(\alpha_n R) \left(\beta_m \cos(\beta_m Z)+ (h/D) \sin(\beta_m Z)+\beta_m\right)\right. \nonumber\\
    &\quad{}+ \left.R\,J_0(\alpha_n R)\left[\frac{h}{D \beta_m} + \sin(\beta_m Z) - \frac{h}{D \beta_m}\cos(\beta_m Z)\right]\right]\frac{1}{D(\alpha_n^2+\beta_m^2)}\left[1-\exp(-D(\alpha_n^2+\beta_m^2)t)\right] \label{eq:analysis:total_flux}
\end{align}
\begin{tabular}{p{.975\linewidth}}
\\\hline\\[-2.0em]
\end{tabular}
\end{figure*}

There exist two computational challenges when evaluating \eqref{eq:analysis:total_flux} numerically; (i) the computation of the infinite sum and (ii) the computation of the eigenvalues defined by implicit expressions \eqref{eq:thm:radial_eigv} and \eqref{eq:thm:axial_eigv}.
To overcome (i), the infinite sums are truncated to finite sums in the numerical evaluation (in this paper, we consider the first 250 terms in each sum).
In order to make the truncation error as small as possible, the terms with the {\em smallest} eigenvalues $\alpha_n$, $\beta_m$ are considered for each sum, since these contribute most to the total value of $F(t)$ as $t$ increases.
With respect to (ii), we note that finding the solutions to \eqref{eq:thm:radial_eigv} is simplified by the fact that $J_0(\cdot)$ and $J_1(\cdot)$ oscillate with the same frequency and exactly one solution exists for each oscillation.
Furthermore, the left-hand side in \eqref{eq:thm:axial_eigv} is periodic and monotonically increasing from $-\infty$ to $\infty$ in each interval $[(k-1/2)\pi/Z, (k+1/2)\pi/Z]$, $k\in\mathbb{Z}$, where $\mathbb{Z}$ denotes the set of integers, while the right-hand side in \eqref{eq:thm:axial_eigv} is monotonically decreasing for all $\beta>0$, except at the pole $\beta=h$.
Hence, for each period of $\tan(\beta Z)$ there exists exactly one solution to \eqref{eq:thm:axial_eigv} except for the interval containing the pole $\beta=h$ in which there exists one additional solution.
Finally, the eigenvalues $\alpha_n$ and $\beta_m$ can be computed offline and then be re-used in each evaluation of \eqref{eq:analysis:total_flux} for different values of $t$.
In summary, the computational effort for evaluating \eqref{eq:analysis:total_flux} is minor.

\section{Evaluation}\label{sec:evaluation}
Figure~\ref{fig:drug_release} shows $F(t)$, i.e., the cumulative fraction of drug released from the \ac{BNC} fleece at time $t$ relative to the total amount of drug initially loaded into the \ac{BNC}. We show the experimental release profiles (markers) and the analytical predictions (solid lines) for three different scenarios:
\begin{itemize}
    \item \textbf{Green:} Drug release from uncoated BNC fleeces.
    \item \textbf{Blue:} Drug release from BNC fleeces coated for 20 minutes.
    \item \textbf{Orange:} Drug release from BNC fleeces coated for 30 minutes.
\end{itemize}

In addition, simulation results are provided for the cases that the diffusion coefficient $D$ is increased or decreased by $20\%$, respectively, for the scenario with 20 minute coating duration.
For all experimental results reported in Figure~\ref{fig:drug_release}, the replicate count is $n=3$.

\subsection{Parameter Value Setting}\label{sec:parameter_values}
The fitting procedure proceeds in two stages.
First, the fleece radius $R$ and height $Z$ are fixed to experimentally measured values, namely $R=7.6\times10^{-3}$ $\mathrm{m}$ and $Z=4.4\times10^{-3}$ $\mathrm{m}$.
Then, the diffusion coefficient $D$ is determined as $D=2.3649\times10^{-10}$ $\textrm{m}^2/\textrm{s}$ from the uncoated dataset by setting the coating thickness $l$ to a very small value in the model, so that the boundary conditions approach Dirichlet-type behavior (no diffusive resistance at the fleece boundaries).
The obtained value of $D$ is in line with previous studies of diffusion-controlled drug release from a \ac{BNC} fleece \cite{lotter_microparticle-based_2023}.

Second, $l$ is set to $l_{20} = 1.25\times10^{-4}$ $\mathrm{m}$, as measured experimentally from images of 20~min coated samples, and the effective diffusion coefficient in the coating, $D_\mathrm{c}$, is fitted for the 20~min coating dataset as $D_\mathrm{c}=3.3417\times10^{-11}$ $\textrm{m}^2/\textrm{s}$. Finally, using $D_\mathrm{c}=3.3417\times10^{-11}$ $\textrm{m}^2/\textrm{s}$, the coating thickness $l_{30} = 3.14\times10^{-4}$ $\mathrm{m}$ is fitted for the 30~min dataset. All measured and fitted parameter values are summarized in Table~\ref{table:summary_values}.

\begin{table}[!t]
    \centering
    \caption{Measured and fitted parameter values.}\label{table:summary_values}
    \def\arraystretch{1.4}
    \resizebox{\columnwidth}{!}{
    \begin{tabular}{|l|c|c|c|c|}
        \hline
        Parameter & Variable & Value & Measured & Fitted \\
        \hline\hline

        Fleece radius  &$R$ & $7.6\times10^{-3}$ $\mathrm{m}$ & \checkmark &  \\ \hline
        Fleece height & $Z$ & $4.4\times10^{-3}$ $\mathrm{m}$ & \checkmark &  \\ \hline

        Diffusion coefficient of \ac{BNC} & $D$ & $2.3649\times10^{-10}$ $\textrm{m}^2/\textrm{s}$ &  & \checkmark \\ \hline

        Coating times & -- & 20 min, 30 min & \checkmark &  \\ \hline

        Coating thickness after 20 min & $l_{20}$ &
            $1.25\times10^{-4}$ $\mathrm{m}$ & \checkmark & \\ \hline

        Coating thickness after 30 min & $l_{30}$ &
            $3.14\times10^{-4}$ $\mathrm{m}$ &  & \checkmark \\ \hline
            
        Coating diffusion coefficient & $D_\mathrm{c}$ & $3.3417\times10^{-11}$ $\textrm{m}^2/\textrm{s}$ &  & \checkmark \\ \hline
    \end{tabular}
    }
\end{table}

\subsection{Results and Discussion}\label{sec:results_and_discussion}
The following observations can be made:
\begin{enumerate}
    \item The analytical model reproduces the experimental results with excellent accuracy across all scenarios (mean squared errors: $2.0235\times10^{-4}$ for no coating, $1.2938\times 10^{-4}$ for 20 min coating, $2.3411\times 10^{-4}$ for 30 min coating), confirming that the simplifying assumptions of uniform initial drug distribution and thin-layer coating approximation are adequate for capturing the main release dynamics.
    \item The polymer coating markedly slows the release kinetics, resulting in an extended drug release compared to the uncoated case. This behavior is consistent with the additional diffusive resistance imposed by the coating, as modeled through the boundary conditions \eqref{eq:system_model:radial_boundary}-\eqref{eq:system_model:axial_boundary_up}.
    \item Increasing the coating time from 20~min to 30~min further extends the release, consistent with the thicker coating shell produced by longer polymer deposition.
    \item The fitted $l$ for the 30~min coating is larger than for the 20~min coating, in agreement with the expectation that increased coating thickness leads to elongated diffusion pathways.
    \item Increasing/decreasing $D$ by $20\%$ leads to only minor changes in the simulation results, indicating relatively low sensitivity of the model with respect to the drug mobility inside the fleece.
\end{enumerate}

Beyond these direct observations, several broader implications emerge from the results.
First, the ability of the model to accurately match data for both coated and uncoated samples without re-fitting geometry or bulk diffusivity parameters underscores its predictive potential.
Once the intrinsic BNC diffusivity $D$ has been determined from one reference measurement, the model can be used to predict release profiles for arbitrary coating permeabilities.
This could significantly reduce experimental workload when screening different coating protocols.

Second, the strong correlation between coating time and the extracted $l$ parameter provides a quantitative link between an easily controllable manufacturing variable and the drug retention within the coating.
This relationship can be exploited in formulation design: for example, choosing a coating protocol to achieve a desired release duration without changing the core BNC properties.

Third, the observed agreement between theory and experiment suggests that additional complexities, such as swelling of the BNC matrix or drug–polymer interactions, either have negligible effects in the studied time frame or can be effectively incorporated into an ``effective'' diffusivity parameter.
Nevertheless, deviations might become significant for highly hydrophobic drugs or in media that alter coating integrity, which points toward natural extensions of the current model.

Finally, the extended release observed in the coated samples demonstrates the suitability of polymer-coated BNC fleeces for gastroretentive drug delivery applications where long gastric residence and controlled release are critical, such as for narrow-absorption-window drugs or local gastric treatments.

\begin{figure}
    \centering
    \includegraphics[width=0.99\linewidth]{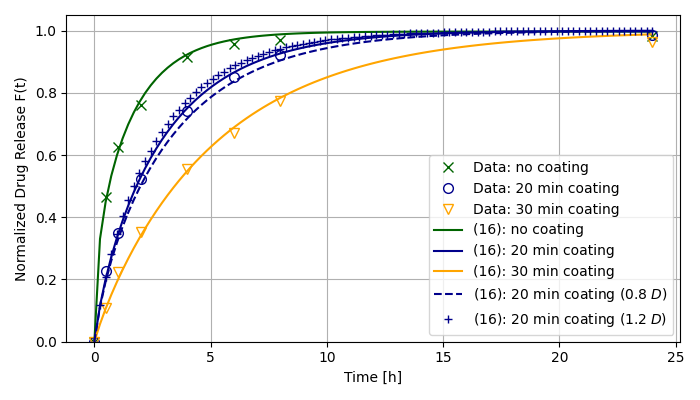}
    \caption{Data from drug release experiments and analytical model \eqref{eq:analysis:total_flux}.}\label{fig:drug_release}
\end{figure}

\section{Conclusion}\label{sec:conclusion}
We have presented a novel analytical model for drug release from polymer-coated BNC fleeces, derived an explicit analytical solution to the underlying boundary value problem, and validated the model with experimental data.
The results show that the model can accurately capture the influence of coating thickness and permeability on drug release kinetics.

Beyond explaining the observed experimental trends, the proposed framework provides a computationally efficient tool for predicting release profiles for different BNC geometries and coating properties.
This capability can inform the rational design and optimization of future BNC-based gastroretentive drug delivery systems, potentially accelerating the translation of such systems from laboratory concepts to clinically viable therapeutics.

\appendices
\section{Proof of Theorem~\ref{thm:concentration}}
We seek a separated solution of
\begin{align}
\frac{1}{D}\frac{\partial c}{\partial t} = \frac{1}{r}\frac{\partial}{\partial r}\left(r \frac{\partial c}{\partial r}\right) +  \frac{\partial^2 c}{\partial z^2},\quad c(r,z,0)=c_0,
\end{align}

on \([0,R]\times[0,Z]\) with the Robin conditions stated in \eqref{eq:system_model:radial_boundary}--\eqref{eq:system_model:inner_bd_condition}.

Let
\begin{align}
c(r,z,t)=\mathcal{T}(t)\mathcal{R}(r)\mathcal{Z}(z).
\end{align}
Substituting and separating gives
\begin{align}
&\frac{\mathcal{T}'} {D\mathcal{T}}= -(\alpha^2+\beta^2),\,\,
\frac1r(r\mathcal{R}')'+\alpha^2\mathcal{R}=0,\,\,\nonumber\\ &\mathcal{Z}''+\beta^2 \mathcal{Z}=0,
\end{align}
where $f'$ and $f''$ denote the first and second order derivative of function $f$.
Hence,
\begin{align}
\mathcal{T}(t)=e^{-D(\alpha^2+\beta^2)t}.
\end{align}

\paragraph{Radial eigenproblem}
The regular solution of \(\frac1r(r\mathcal{R}')'+\alpha^2\mathcal{R}=0\) is
\begin{align}
\mathcal{R}(r)=J_0(\alpha r).
\end{align}
The boundary condition \(-D\mathcal{R}'(R)=h\mathcal{R}(R)\) yields
\begin{align}
\gamma J_1(\gamma)=\frac{hR}{D} J_0(\gamma),\qquad \gamma=\alpha R,
\end{align}
whose positive roots are \(\gamma_n\). Thus
\begin{align}
\alpha_n=\gamma_n/R,\qquad \mathcal{R}_n(r)=J_0(\alpha_n r).
\end{align}

\paragraph{Axial eigenproblem}
Solving \(\mathcal{Z}''+\beta^2 \mathcal{Z}=0\) with
\begin{align}
D\mathcal{Z}'(0)=h\mathcal{Z}(0),\qquad -D\mathcal{Z}'(Z)=h\mathcal{Z}(Z),
\end{align}
leads to the eigenvalue equation
\begin{align}\label{eq:app:eig_Z}
\tan(\beta Z)=\frac{2h\beta}{D(\beta^2-h^2/D^2)}.
\end{align}
Let \(\beta_m\) be the positive roots of \eqref{eq:app:eig_Z}, and define the corresponding eigenfunction
\begin{align}
\mathcal{Z}_m(z)=\beta_m\cos(\beta_m z)+\frac{h}{D}\sin(\beta_m z).
\end{align}

\paragraph{Mode expansion}
The separated modes are
\begin{align}
\Phi_{nm}(r,z,t)
=J_0(\alpha_n r)\,\mathcal{Z}_m(z)\,e^{-D(\alpha_n^2+\beta_m^2)t}.
\end{align}
Hence, we have
\begin{align}
c(r,z,t)=\sum_{n,m}A_{nm}\Phi_{nm}(r,z,t),
\end{align}
with coefficients $A_{nm}$.
\paragraph{Projection}
With inner product
\begin{align}
\langle f,g\rangle=\int_0^Z\int_0^R f g\,2\pi r\,\mathrm{d}r\,\mathrm{d}z,
\end{align}
the coefficients are
\begin{align}
A_{nm} = \frac{\mathcal I_{nm}}{2\pi N(\alpha_n)M(\beta_m)}\,,
\end{align}
where $\mathcal I_{nm} = \langle c_0, J_0(\alpha_n r)\mathcal{Z}_m(z)\rangle$ and $2\pi N(\alpha_n)M(\beta_m) = \langle J_0(\alpha_n r)\mathcal{Z}_m(z),J_0(\alpha_n r)\mathcal{Z}_m(z)\rangle$.

Because \(c_0=\tfrac{1}{ZR^2\pi}\),
\begin{align}
\int_0^R c_0 J_0(\alpha_n r)\,2\pi r\,\mathrm{d}r
=\frac{2}{ZR\alpha_n}J_1(\alpha_n R),
\end{align}
and
\begin{align}
\int_0^Z \mathcal{Z}_m(z)\,\mathrm{d}z
=\sin(\beta_m Z)+\frac{h}{D\beta_m}(1-\cos(\beta_m Z)).
\end{align}
Thus,
\begin{align}
\mathcal I_{nm}
=&\frac{2}{ZR\alpha_n}J_1(\alpha_n R)\nonumber\\
&\quad \Big[\sin(\beta_m Z)+\frac{h}{D\beta_m}(1-\cos(\beta_m Z))\Big].
\end{align}

\paragraph{Norms}
Using standard Bessel identities and the Robin eigenvalue relation,
\begin{align}
N(\alpha_n)=\frac{J_0^2(\alpha_n R)}{2} R^2
\frac{(h/D)^2+\alpha_n^2}{\alpha_n^2}.
\end{align}
The direct evaluation of \(\int_0^Z \mathcal{Z}_m^2(z) \mathrm{d}z\) gives
\begin{align}
M(\beta_m)&=\frac12\Big[ \left(\beta_m^2+(h/D)^2\right)
\Big(Z+\frac{h/D}{\beta_m^2+(h/D)^2}\Big) \nonumber\\
&\quad +\frac{h}{D} \Big].
\end{align}

\paragraph{Final form}
Combining all expressions and simplifying yields exactly the coefficient formula \eqref{eq:thm:coefficients} and definitions \eqref{eq:thm:radial_eigv}, \eqref{eq:thm:axial_eigv}, \eqref{eq:thm:radial_norm}, \eqref{eq:thm:axial_norm}. This completes the proof.
\bibliographystyle{IEEEtran}
\bibliography{bibtex/refs}

\begin{thebibliography}{10}
\providecommand{\url}[1]{#1}
\csname url@samestyle\endcsname
\providecommand{\newblock}{\relax}
\providecommand{\bibinfo}[2]{#2}
\providecommand{\BIBentrySTDinterwordspacing}{\spaceskip=0pt\relax}
\providecommand{\BIBentryALTinterwordstretchfactor}{4}
\providecommand{\BIBentryALTinterwordspacing}{\spaceskip=\fontdimen2\font plus
\BIBentryALTinterwordstretchfactor\fontdimen3\font minus \fontdimen4\font\relax}
\providecommand{\BIBforeignlanguage}[2]{{%
\expandafter\ifx\csname l@#1\endcsname\relax
\typeout{** WARNING: IEEEtran.bst: No hyphenation pattern has been}%
\typeout{** loaded for the language `#1'. Using the pattern for}%
\typeout{** the default language instead.}%
\else
\language=\csname l@#1\endcsname
\fi
#2}}
\providecommand{\BIBdecl}{\relax}
\BIBdecl

\bibitem{adepu_controlled_2021}
S.~Adepu and S.~Ramakrishna, ``Controlled {Drug} {Delivery} {Systems}: {Current} {Status} and {Future} {Directions},'' \emph{Molecules}, vol.~26, no.~19, Sep. 2021.

\bibitem{chude-okonkwo_molecular_2017}
U.~A.~K. Chude-Okonkwo, R.~Malekian, B.~T. Maharaj, and A.~V. Vasilakos, ``Molecular {Communication} and {Nanonetwork} for {Targeted} {Drug} {Delivery}: {A} {Survey},'' \emph{IEEE Commun. Surv. Tut.}, vol.~19, no.~4, pp. 3046--3096, Oct. 2017.

\bibitem{chahibi_molecular_2013}
Y.~Chahibi, M.~Pierobon, S.~O. Song, and I.~F. Akyildiz, ``A {Molecular} {Communication} {System} {Model} for {Particulate} {Drug} {Delivery} {Systems},'' \emph{IEEE Trans. Biomed. Eng.}, vol.~60, no.~12, pp. 3468--3483, Dec. 2013.

\bibitem{chen_modeling_2017}
Y.~Chen, Y.~Zhou, R.~Murch, and P.~Kosmas, ``Modeling {Contrast}-{Imaging}-{Assisted} {Optimal} {Targeted} {Drug} {Delivery}: {A} {Touchable} {Communication} {Channel} {Estimation} and {Waveform} {Design} {Perspective},'' \emph{IEEE Trans. NanoBiosci.}, vol.~16, no.~3, pp. 203--215, Apr. 2017.

\bibitem{chude-okonkwo_information-theoretic_2020}
U.~A.~K. Chude-Okonkwo, B.~T. Maharaj, A.~V. Vasilakos, and R.~Malekian, ``Information-{Theoretic} {Model} and {Analysis} of {Molecular} {Signaling} in {Targeted} {Drug} {Delivery},'' \emph{IEEE Trans. NanoBiosci.}, vol.~19, no.~2, pp. 270--284, Apr. 2020.

\bibitem{zhao_release_2021}
Q.~Zhao, M.~Li, and L.~Lin, ``Release {Rate} {Optimization} in {Molecular} {Communication} for {Local} {Nanomachine}-{Based} {Targeted} {Drug} {Delivery},'' \emph{IEEE Trans. NanoBiosci.}, vol.~20, no.~4, pp. 396--405, Oct. 2021.

\bibitem{femminella_molecular_2015}
M.~Femminella, G.~Reali, and A.~V. Vasilakos, ``A {Molecular} {Communications} {Model} for {Drug} {Delivery},'' \emph{IEEE Trans. NanoBiosci.}, vol.~14, no.~8, pp. 935--945, Oct. 2015.

\bibitem{salehi_diffusion-based_2019}
S.~Salehi, N.~S. Moayedian, and E.~Alarcón, ``Diffusion-{Based} {Molecular} {Communication} {Channel} in {Presence} of a {Probabilistic} {Absorber}: {Single} {Receptor} {Model} and {Congestion} {Analysis},'' \emph{IEEE Trans. NanoBiosci.}, vol.~18, no.~1, pp. 84--92, Jan. 2019.

\bibitem{zhao_adaptive_2021}
Q.~Zhao and L.~Lin, ``Adaptive {Release} {Rate} in {Drug} {Delivery} {Based} on {Mobile} {Molecular} {Communication},'' in \emph{{Proc.} {IEEE} {Wireless} {Commun}. {Netw}. {Conf}. ({WCNC})}, Nanjing, China, 2021, pp. 1--6.

\bibitem{sun2025bio}
Y.~Sun, H.~Tan, and Y.~Chen, ``{A Bio-Nano Systems Interconnection Hierarchical Network Model for Targeted Drug Delivery},'' \emph{IEEE Internet Things J.}, vol.~12, no.~21, pp. 44\,867--44\,881, Nov. 2025.

\bibitem{wang2021optimal}
D.~Wang, Y.~Sun, Y.~Xiao, and Y.~Chen, ``{An Optimal Strategy for Individualized Drug Delivery Therapy: A Molecular Communication Inspired Waveform Design Perspective},'' in \emph{Proc. IEEE Eng. Med. Biol. Soc.}, 2021, pp. 866--869.

\bibitem{veletic_modeling_2020}
M.~Veletić, M.~T. Barros, H.~Arjmandi, S.~Balasubramaniam, and I.~Balasingham, ``Modeling of {Modulated} {Exosome} {Release} {From} {Differentiated} {Induced} {Neural} {Stem} {Cells} for {Targeted} {Drug} {Delivery},'' \emph{IEEE Trans. NanoBiosci.}, vol.~19, no.~3, pp. 357--367, Jul. 2020.

\bibitem{damrath_optimization_2024}
M.~Damrath, M.~Veletić, H.~K. Rudsari, and I.~Balasingham, ``Optimization of {Extracellular} {Vesicle} {Release} for {Targeted} {Drug} {Delivery},'' \emph{IEEE Trans. NanoBiosci.}, vol.~23, no.~1, pp. 109--117, Jan. 2024.

\bibitem{rudsari_targeted_2021}
H.~K. Rudsari, M.~Veletić, J.~Bergsland, and I.~Balasingham, ``Targeted {Drug} {Delivery} for {Cardiovascular} {Disease}: {Modeling} of {Modulated} {Extracellular} {Vesicle} {Release} {Rates},'' \emph{IEEE Trans. NanoBiosci.}, vol.~20, no.~4, pp. 444--454, Oct. 2021.

\bibitem{lotter_experimental_2023}
S.~Lotter \emph{et~al.}, ``Experimental {Research} in {Synthetic} {Molecular} {Communications} – {Part} {I},'' \emph{IEEE Nanotechnol. Mag.}, vol.~17, no.~3, pp. 42--53, Jun. 2023.

\bibitem{potzinger_bacterial_2017}
Y.~Pötzinger, D.~Kralisch, and D.~Fischer, ``Bacterial {Nanocellulose}: {The} {Future} of {Controlled} {Drug} {Delivery}?'' \emph{Therapeutic Del.}, vol.~8, no.~9, pp. 753--761, Aug. 2017.

\bibitem{iguchi_bacterial_2000}
M.~Iguchi, S.~Yamanaka, and A.~Budhiono, ``{Bacterial Cellulose – A Masterpiece of Nature's Arts},'' \emph{J. Mat. Sci.}, vol.~35, no.~2, pp. 261--270, Jan. 2000.

\bibitem{jain_theoretical_2022}
A.~Jain, S.~McGinty, G.~Pontrelli, and L.~Zhou, ``{Theoretical Model for Diffusion-Reaction Based Drug Delivery from a Multilayer Spherical Capsule},'' \emph{Int. J. Heat Mass Transfer}, vol. 183, Feb. 2022.

\bibitem{hahn_heat_2012}
D.~W. Hahn and M.~N. Özisik, \emph{\BIBforeignlanguage{en}{Heat {Conduction}}}.\hskip 1em plus 0.5em minus 0.4em\relax John Wiley \& Sons, Aug. 2012.

\bibitem{lotter_microparticle-based_2023}
S.~Lotter \emph{et~al.}, ``Microparticle-{Based} {Controlled} {Drug} {Delivery} {Systems}: {From} {Experiments} to {Statistical} {Analysis} and {Design},'' in \emph{Proc. {IEEE} {Glob.} {Commun}. {Conf}.}, Dec. 2023, pp. 1167--1172.

\end{thebibliography}

\end{document}